\documentclass[preprint,12pt]{elsarticle}

\usepackage{amsmath}
\usepackage{amssymb}
\usepackage{amsthm}
\usepackage{algorithm}
\usepackage{algpseudocode}
\usepackage{hyperref}

\journal{Theoretical Computer Science}
\biboptions{sort&compress}

\newtheorem{definition}{Definition}
\newtheorem{lemma}{Lemma}
\newtheorem{theorem}{Theorem}

\newcommand{\OPT}{\operatorname{OPT}}
\newcommand{\MSPEC}{\textsc{MSPEC}}

\begin{document}

\begin{frontmatter}

\title{An Exact Polynomial-Time Algorithm for the
Minimum Shared-Power Edge Cut Problem}

\author[inst1]{Quanlong Niu}

\author[inst2]{Wei Liang}

\author[inst2]{Riheng Jia\corref{cor1}}
\ead{rihengjia@zjnu.edu.cn}

\author[inst3]{Zhao Zhang}

\cortext[cor1]{Corresponding author}

\affiliation[inst1]{
    organization={School of Big Data and Artificial Intelligence, Fuyang University of Technology},
    city={Fuyang},
    postcode={236000},
    state={Anhui},
    country={China}
}

\affiliation[inst2]{
    organization={School of Computer Science and Technology, Zhejiang Normal University},
    city={Jinhua},
    postcode={321004},
    state={Zhejiang},
    country={China}
}

\affiliation[inst3]{
    organization={School of Mathematical Sciences, Zhejiang Normal University},
    city={Jinhua},
    postcode={321004},
    state={Zhejiang},
    country={China}
}

\begin{abstract}
The Minimum Shared-Power Edge Cut (MSPEC) problem seeks to
disconnect two designated terminals with minimum total power by
assigning powers to graph vertices and deleting every edge for which
the sum of the powers at its endpoints reaches its edge weight.
Previous work gave a fully polynomial-time approximation scheme for
this problem, while its exact computational complexity remained open.
We prove that general MSPEC with nonnegative edge weights can be
solved exactly in polynomial time. Specifically, we first use linear
programming duality for each fixed cut to reformulate MSPEC as
minimizing the maximum terminal-relaxed matching weight over all
$s$-$t$ cuts. Then, we construct a terminal-copy extended graph that
transforms the maximum-weight matching model with unconstrained
terminal capacities for a fixed cut into an equivalent ordinary
maximum-weight matching model. Next, we extend the submodularity of
the unweighted maximum-matching cut function to nonnegative edge
weights and prove that the restriction obtained by fixing the terminal
copies on prescribed sides remains submodular. Consequently, MSPEC
reduces to exact submodular function minimization. Once an optimal cut
is obtained, a corresponding optimal power assignment can be
recovered. This result resolves the open question concerning the exact
computational complexity of general MSPEC.
\end{abstract}


\end{frontmatter}

\section{Introduction}
\label{sec:introduction}

Cut problems are fundamental models for analyzing network
connectivity. In a conventional minimum cut, costs are assigned to
edges, and hence the deletion costs of distinct cut edges are
independent. In wireless sensor networks and other settings in which
edge states are jointly determined by their endpoints, however,
changing the sensing capability of one sensor may simultaneously alter
its connections to several other sensors. Barrier-coverage studies
commonly measure barrier strength by the number of failed
sensors~\cite{kumar2007barrier}. Measuring barrier strength by
adjusting sensing ranges instead requires a model that captures the
coupling between resource allocation at vertices and the deletion of
cut edges~\cite{cabello2020mspec}.

Motivated by this setting, Cabello et
al.~\cite{cabello2020mspec} introduced the Minimum Shared-Power Edge
Cut problem. A nonnegative power is assigned to each nonterminal
vertex, and an edge is deleted when the sum of the powers at its
endpoints reaches its weight. The objective is to disconnect two
designated terminals while minimizing the total assigned power. This
model generalizes the Minimum Shrinkage problem for rectangular
barrier coverage and differs from the conventional minimum cut, whose
objective is the sum of the weights of the cut edges. Previous work
developed a fully polynomial-time approximation scheme for MSPEC and
polynomial-time algorithms for special cases such as uniform edge
weights and restricted power domains. Whether the general
edge-weighted problem can be solved exactly in polynomial time was
left as a central open question. Cabello and Colin de Verdi\`ere
subsequently proved weak NP-hardness for the general geometric Minimum
Barrier Shrinkage problem~\cite{cabello2020hardness}. Cabello and
Mulzer also studied minimum $s$-$t$ vertex cuts in disk intersection
graphs, where the objective is to remove as few disk vertices as
possible in order to disconnect $s$ and
$t$~\cite{cabello2021cuts}.

For a fixed $s$-$t$ cut, the minimum required power is described by a
vertex-cover-type linear program on the bipartite cut graph, whose
dual is a maximum-weight matching problem with unconstrained terminal
capacities~\cite{cabello2020mspec}. Thus, the key difficulty is to
find, among all $s$-$t$ cuts, one whose associated terminal-relaxed
maximum-weight matching has minimum value. Separately, S{\ae}ther and
Telle proved that the unweighted maximum-matching cut function is
submodular and used this property in the study of graph width
parameters~\cite{saether2016between}.

This paper combines the weighted matching formulation of a fixed cut
with the submodular structure of the matching-cut function, thereby
avoiding a discretized enumeration of cuts or vertex powers. First,
using the primal--dual relationship for a fixed cut, we formulate
MSPEC exactly as the minimization of a terminal-relaxed maximum-weight
matching-cut function. Second, we introduce a dedicated copy for every
terminal-incident edge, transforming the reusable-terminal matching
constraint into an ordinary matching constraint in an extended graph.
We then prove that the resulting maximum-weight matching-cut function
is submodular for nonnegative edge weights. Furthermore, we show that
the restriction obtained by fixing all copies of terminal $s$ on the
source side and all copies of terminal $t$ on the sink side remains
submodular. Finally, we use a fixed-flow maximum-cost flow algorithm
to evaluate the matching value of each fixed
cut~\cite{ahuja1993network}, employ an exact submodular function
minimization algorithm to obtain an optimal
cut~\cite{jiang2024convex}, and recover a minimum power assignment by
solving the fixed-cut linear
program~\cite{schrijver2003combinatorial}. The resulting method is an
exact polynomial-time algorithm for general MSPEC and also yields an
exact algorithm for Minimum Shrinkage in rectangular barrier coverage.

\section{Preliminaries and Problem Formulation}
\label{sec:preliminaries}

\subsection{Problem Definition}

An instance of general MSPEC consists of an undirected graph
\[
G=(U\cup\{s,t\},E),
\]
where $s$ and $t$ are the terminals, $U$ is the set of nonterminal
vertices, and every edge $e=uv\in E$ has a nonnegative weight $w_e$.
A nonnegative power $p_v$ is assigned to every $v\in U$, while
\[
p_s=p_t=0.
\]
An edge $uv$ is deleted whenever
\[
p_u+p_v\geq w_{uv}.
\]
We assume without loss of generality that $st\notin E$. Indeed, if
$st\in E$ and $w_{st}>0$, then the instance is infeasible because
$p_s=p_t=0$ and the edge $st$ cannot be deleted. If $w_{st}=0$, then
$st$ already satisfies the deletion condition at zero power and can
be removed in preprocessing.

The objective is to disconnect $s$ and $t$ after the edge deletions
while minimizing
\[
\sum_{v\in U}p_v.
\]
The optimal objective value is denoted by $\OPT_{\MSPEC}$.

\subsection{Power Allocation on a Fixed Cut}

Let
\[
V=U\cup\{s,t\}.
\]
For any $S\subseteq V$, let
\[
\overline{S}=V\setminus S
\]
and define the corresponding cut-edge set by
\[
\delta_G(S)=\{uv\in E:u\in S,\ v\in\overline{S}\}.
\]
Every $s$-$t$ cut can be represented by
\[
S(X)=\{s\}\cup X,
\qquad X\subseteq U.
\]
We write
\[
C(X)=\delta_G(S(X))
\]
for its cut-edge set.

Once $S(X)$ is fixed, all edges in $C(X)$ must be deleted. Hence every
$uv\in C(X)$ must satisfy
\[
p_u+p_v\geq w_{uv}.
\]
Only the powers of vertices in $U$ are decision variables; the
terminal powers remain fixed at zero. The minimum-power assignment
problem for the fixed cut $S(X)$ is therefore
\begin{equation}
\begin{aligned}
(P_X)\qquad
\min\quad & \sum_{v\in U}p_v \\
\text{s.t.}\quad
& p_u+p_v\geq w_{uv}, && \forall uv\in C(X),\\
& p_v\geq0, && \forall v\in U,\\
& p_s=p_t=0.
\end{aligned}
\label{eq:fixed-primal}
\end{equation}
Its optimal objective value is denoted by $\OPT(P_X)$.

\subsection{Dual of the Fixed-Cut Problem}

Introduce one dual variable $y_e\geq0$ for each cut edge
$e\in C(X)$. The dual of~\eqref{eq:fixed-primal} is
\begin{equation}
\begin{aligned}
(D_X)\qquad
\max\quad & \sum_{e\in C(X)}w_e y_e \\
\text{s.t.}\quad
& \sum_{\substack{e\in C(X)\\e\ni v}}y_e\leq1,
&& \forall v\in U,\\
& y_e\geq0, && \forall e\in C(X).
\end{aligned}
\label{eq:fixed-dual}
\end{equation}
The dual constraints correspond only to nonterminal vertices. This is
because $p_s$ and $p_t$ are fixed at zero and are not decision
variables of the primal program; consequently, there are no capacity
constraints for the terminals in~$(D_X)$.

The edge set $C(X)$ induces a bipartite graph with bipartition
\[
\bigl(S(X),\overline{S(X)}\bigr).
\]
The vertex--edge incidence matrix of a bipartite graph is totally
unimodular. The constraint matrix of~$(D_X)$ is obtained from this
incidence matrix by deleting the rows corresponding to $s$ and $t$.
Since every submatrix of a totally unimodular matrix remains totally
unimodular and the right-hand side of~$(D_X)$ is integral, the
feasible polyhedron of~$(D_X)$ has integral extreme points. Hence
$(D_X)$ has an integral optimal solution. Moreover, every cut edge is
incident to at least one nonterminal vertex, and each nonterminal
vertex has capacity one. Therefore, an integral optimal solution
satisfies
\[
y_e\in\{0,1\},
\qquad \forall e\in C(X).
\]
Here $y_e=1$ means that edge $e$ is selected, whereas $y_e=0$ means
that it is not. Each nonterminal vertex is incident to at most one
selected edge, while either terminal may be incident to multiple
selected edges.

\begin{definition}[Terminal-relaxed matching]
Given $C(X)$, an edge set $M\subseteq C(X)$ is a
\emph{terminal-relaxed matching} if
\[
\deg_M(v)\leq1,
\qquad \forall v\in U,
\]
while the numbers of selected edges incident to $s$ and $t$ are
unrestricted.
\end{definition}

Define the maximum weight of a terminal-relaxed matching on the fixed
cut $S(X)$ by
\begin{equation}
\begin{aligned}
\mu(X)=\max\Bigl\{\sum_{e\in M}w_e:
&\ M\subseteq C(X),\\[-1mm]
&\ M\text{ is a terminal-relaxed matching}\Bigr\}.
\end{aligned}
\label{eq:terminal-relaxed-matching}
\end{equation}
By the integrality of~$(D_X)$,
\[
\OPT(D_X)=\mu(X).
\]
By strong linear-programming duality,
\[
\OPT(P_X)=\OPT(D_X).
\]
Consequently, for every $X\subseteq U$,
\begin{equation}
\OPT(P_X)=\mu(X).
\label{eq:fixed-cut-value}
\end{equation}

\subsection{Fixed-Cut Characterization of MSPEC}

The following fixed-cut characterization follows from Cabello et
al.~\cite{cabello2020mspec}.

\begin{lemma}
Let $P_X$ be the minimum-power assignment problem for the cut
$S(X)=\{s\}\cup X$, where $X\subseteq U$. Then
\begin{equation}
\OPT_{\MSPEC}
=
\min_{X\subseteq U}\OPT(P_X).
\label{eq:fixed-cut-characterization}
\end{equation}
\end{lemma}

Combining~\eqref{eq:fixed-cut-value}
and~\eqref{eq:fixed-cut-characterization} gives
\begin{equation}
\OPT_{\MSPEC}
=
\min_{X\subseteq U}\mu(X).
\label{eq:mspec-terminal-relaxed}
\end{equation}

\section{Terminal-Copy Extended Graph}
\label{sec:extended-graph}

Define the sets of terminal-incident edges by
\[
E_s=\{su\in E:u\in U\},
\qquad
E_t=\{ut\in E:u\in U\}.
\]
For every $e\in E_s$, introduce a dedicated copy $s_e$ of terminal
$s$; similarly, for every $e\in E_t$, introduce a dedicated copy
$t_e$ of terminal $t$. Let
\[
C_s=\{s_e:e\in E_s\},
\qquad
C_t=\{t_e:e\in E_t\}.
\]
We construct the extended graph
\[
H=(W,F),
\]
where
\[
W=U\cup C_s\cup C_t.
\]
Let
\[
E_U=\{uv\in E:u,v\in U\}.
\]
The edge set of $H$ is
\begin{equation}
F
=
E_U
\cup\{s_eu:su\in E_s\}
\cup\{ut_e:ut\in E_t\}.
\label{eq:extended-edge-set}
\end{equation}
Every edge of $H$ inherits the weight of its corresponding edge in
$G$.

For every $X\subseteq U$, define one side of the corresponding cut in
$H$ by
\[
\widehat{S}(X)=C_s\cup X,
\]
whose complement is
\[
W\setminus\widehat{S}(X)=C_t\cup(U\setminus X).
\]
For every $Z\subseteq W$, define the ordinary maximum-weight
matching-cut function in $H$ by
\begin{equation}
\begin{aligned}
\mu_H(Z)=\max\Bigl\{\sum_{f\in M}w_f^H:
&\ M\subseteq\delta_H(Z),\\[-1mm]
&\ M\text{ is an ordinary matching}\Bigr\},
\end{aligned}
\label{eq:matching-cut-function}
\end{equation}
where $w_f^H$ denotes the weight of edge $f$ in $H$.

\begin{theorem}[Extended-graph equivalence]
For every $X\subseteq U$, the maximum weight of a terminal-relaxed
matching on $C(X)$ equals the ordinary maximum-weight matching value
on the extended cut $\delta_H(\widehat{S}(X))$; that is,
\begin{equation}
\mu(X)=\mu_H\bigl(\widehat{S}(X)\bigr).
\label{eq:extended-equivalence}
\end{equation}
\end{theorem}

\begin{proof}
There is a natural correspondence between original cut edges and
extended cut edges:
\[
uv\longleftrightarrow uv,
\qquad u,v\in U,
\]
\[
su\longleftrightarrow s_eu,
\qquad e=su\in E_s,
\]
and
\[
ut\longleftrightarrow ut_e,
\qquad e=ut\in E_t.
\]
By the definitions of $S(X)=\{s\}\cup X$ and
$\widehat{S}(X)=C_s\cup X$, an original edge crosses $C(X)$ if and
only if its corresponding extended edge crosses
$\delta_H(\widehat{S}(X))$. Hence the correspondence is a bijection
between the two cut-edge sets.

The correspondence preserves all incidences at nonterminal vertices.
In addition, each original terminal-incident edge uses a distinct
dedicated terminal copy. Therefore, an edge set $M$ is a
terminal-relaxed matching on the original cut if and only if its
corresponding edge set is an ordinary matching on the extended cut.
Finally, every extended edge inherits the weight of its original
edge, so the correspondence preserves total matching weight. The
feasible matchings on the two cuts are thus in a weight-preserving
bijection, which proves~\eqref{eq:extended-equivalence}.
\end{proof}

\section{Submodularity of the Weighted Matching-Cut Function}
\label{sec:submodularity}

We use the following edge-preserving reconstruction property
established by S{\ae}ther and Telle~\cite{saether2016between}. The
symbol $\uplus$ denotes multiset union.

\begin{lemma}[Edge-preserving reconstruction]
Let $H=(W,F)$ be a graph and let $A,B\subseteq W$. Suppose that
$M_{\cap}$ is a matching of $\delta_H(A\cap B)$ and that $M_{\cup}$ is
a matching of $\delta_H(A\cup B)$. Then there exist a matching $N_A$
of $\delta_H(A)$ and a matching $N_B$ of $\delta_H(B)$ such that
\begin{equation}
N_A\uplus N_B
=
M_{\cap}\uplus M_{\cup}.
\label{eq:edge-preserving-reconstruction}
\end{equation}
\end{lemma}

\begin{theorem}
For all $A,B\subseteq W$, the function $\mu_H$ satisfies
\begin{equation}
\mu_H(A)+\mu_H(B)
\geq
\mu_H(A\cap B)+\mu_H(A\cup B).
\label{eq:weighted-submodularity}
\end{equation}
\end{theorem}

\begin{proof}
Let $M_{\cap}$ and $M_{\cup}$ be maximum-weight matchings of the cuts
$\delta_H(A\cap B)$ and $\delta_H(A\cup B)$, respectively. By the
definition of $\mu_H$,
\[
w(M_{\cap})=\mu_H(A\cap B)
\]
and
\[
w(M_{\cup})=\mu_H(A\cup B).
\]
By the edge-preserving reconstruction lemma, there exist a matching
$N_A$ of $\delta_H(A)$ and a matching $N_B$ of $\delta_H(B)$ such that
\[
N_A\uplus N_B=M_{\cap}\uplus M_{\cup}.
\]
Because this reconstruction preserves the edge multiset,
\begin{equation}
w(N_A)+w(N_B)
=
w(M_{\cap})+w(M_{\cup}).
\label{eq:weight-preservation}
\end{equation}
By the definition of a maximum-weight matching,
\[
\mu_H(A)\geq w(N_A)
\]
and
\[
\mu_H(B)\geq w(N_B).
\]
Therefore,
\begin{align*}
\mu_H(A)+\mu_H(B)
&\geq w(N_A)+w(N_B)\\
&=w(M_{\cap})+w(M_{\cup})\\
&=\mu_H(A\cap B)+\mu_H(A\cup B),
\end{align*}
which proves~\eqref{eq:weighted-submodularity}.
\end{proof}

We next restrict the terminal copies to their prescribed sides of the
cut.

\begin{theorem}
Define $f:2^U\rightarrow\mathbb{R}_{\geq0}$ by
\begin{equation}
f(X)=\mu_H(C_s\cup X).
\label{eq:restricted-function}
\end{equation}
Then $f$ is submodular; namely, for all $X,Y\subseteq U$,
\begin{equation}
f(X)+f(Y)
\geq
f(X\cap Y)+f(X\cup Y).
\label{eq:restricted-submodularity}
\end{equation}
\end{theorem}

\begin{proof}
Take arbitrary $X,Y\subseteq U$ and let
\[
A=C_s\cup X,
\qquad
B=C_s\cup Y.
\]
By the distributive laws,
\[
A\cap B=C_s\cup(X\cap Y)
\]
and
\[
A\cup B=C_s\cup(X\cup Y).
\]
The submodularity of $\mu_H$ gives
\begin{align*}
f(X)+f(Y)
&=\mu_H(C_s\cup X)+\mu_H(C_s\cup Y)\\
&=\mu_H(A)+\mu_H(B)\\
&\geq\mu_H(A\cap B)+\mu_H(A\cup B)\\
&=\mu_H(C_s\cup(X\cap Y))
  +\mu_H(C_s\cup(X\cup Y))\\
&=f(X\cap Y)+f(X\cup Y).
\end{align*}
Thus, $f$ is submodular.
\end{proof}

Combining~\eqref{eq:mspec-terminal-relaxed},
\eqref{eq:extended-equivalence}, and~\eqref{eq:restricted-function},
we obtain
\begin{equation}
\OPT_{\MSPEC}
=
\min_{X\subseteq U}\mu(X)
=
\min_{X\subseteq U}\mu_H\bigl(\widehat{S}(X)\bigr)
=
\min_{X\subseteq U}f(X).
\label{eq:mspec-sfm}
\end{equation}

\section{Exact Polynomial-Time Algorithm}
\label{sec:algorithm}

Equation~\eqref{eq:mspec-sfm} reduces MSPEC to submodular function
minimization. The resulting exact algorithm is given below.

\begin{algorithm}[t]
\caption{Exact algorithm for general MSPEC}
\label{alg:exact-mspec}
\begin{algorithmic}[1]
\Require An undirected edge-weighted graph
$G=(U\cup\{s,t\},E)$.
\Ensure An optimal $s$-$t$ cut shore $S^*$ and a corresponding
optimal power assignment $p^*$.
\State Construct the extended graph $H=(W,F)$ by creating a dedicated
terminal copy for every terminal-incident edge.
\State Define the value-oracle function
$f(X)=\mu_H(C_s\cup X)$ for $X\subseteq U$.
\State For each query $f(X)$, solve a maximum-weight matching problem
on the bipartite cut graph induced by $\delta_H(C_s\cup X)$.
\State Use an exact submodular function minimization algorithm to
compute
\[
X^*\in\operatorname*{arg\,min}_{X\subseteq U}f(X).
\]
\State Set $S^*=\{s\}\cup X^*$.
\State Solve the fixed-cut linear program $P_{X^*}$ to obtain an
optimal power assignment $p^*$.
\State \Return $S^*$, $p^*$, and $\sum_{v\in U}p_v^*$.
\end{algorithmic}
\end{algorithm}

\subsection{Value Oracle for \texorpdfstring{$f(X)$}{f(X)}}

For an arbitrary $X\subseteq U$, let
\[
Z=C_s\cup X.
\]
The cut-edge set $\delta_H(Z)$ induces a bipartite graph with
bipartition $(Z,W\setminus Z)$, and $f(X)$ is exactly the
maximum-weight matching value in this graph. It can therefore be
computed exactly in polynomial time. In particular, the matching
problem can be transformed into a fixed-flow maximum-cost flow
problem~\cite{ahuja1993network}.

Create a source $\alpha$ and a sink $\beta$, and construct the flow
network as follows:
\begin{enumerate}
  \item For every $u\in Z$, add an arc $\alpha\to u$ of capacity $1$
  and cost $0$.

  \item For every $v\in W\setminus Z$, add an arc $v\to\beta$ of
  capacity $1$ and cost $0$.

  \item For every cut edge $uv\in\delta_H(Z)$ with $u\in Z$ and
  $v\in W\setminus Z$, add an arc $u\to v$ of capacity $1$ and cost
  $w_{uv}^H$.

  \item Add a bypass arc $\alpha\to\beta$ of capacity $K$ and cost
  $0$, where
  \[
  K=\min\{|Z|,|W\setminus Z|\}.
  \]
\end{enumerate}

Send exactly $K$ units of flow from $\alpha$ to $\beta$ while
maximizing the total cost. One unit of flow along a path
\[
\alpha\to u\to v\to\beta
\]
corresponds to selecting the matching edge $uv$. Flow that does not
select a matching edge is sent through the zero-cost bypass arc. Since
all vertex arcs have unit capacity, the selected cut edges form a
matching. Conversely, every matching can be represented by sending
one unit of flow through the corresponding three-arc path for each
selected edge and routing the remaining flow through the bypass arc.
The objective values are preserved. Hence the optimal maximum-cost
flow value is exactly
\begin{equation}
f(X)=\mu_H(Z).
\label{eq:flow-oracle}
\end{equation}

\begin{theorem}
General MSPEC with nonnegative edge weights can be solved exactly in
polynomial time.
\end{theorem}

\begin{proof}
The restricted matching-cut function $f$ is submodular by
Section~\ref{sec:submodularity}. For every $X\subseteq U$, its value
can be computed exactly in polynomial time by the maximum-cost flow
construction above. An exact submodular function minimization
algorithm therefore returns
\[
X^*\in\operatorname*{arg\,min}_{X\subseteq U}f(X)
\]
in polynomial time. Equation~\eqref{eq:mspec-sfm} implies that the cut
shore
\[
S^*=\{s\}\cup X^*
\]
is optimal for MSPEC. Solving the fixed-cut linear program
$P_{X^*}$ yields a power assignment $p^*$ satisfying
\[
\sum_{v\in U}p_v^*
=
\OPT(P_{X^*})
=
f(X^*)
=
\OPT_{\MSPEC}.
\]
Thus Algorithm~\ref{alg:exact-mspec} returns an optimal cut and an
optimal power assignment in polynomial time.
\end{proof}

\section{Complexity Analysis}
\label{sec:complexity}

Let
\[
n=|U|,
\qquad
m=|E|,
\]
and let
\[
N=|W|,
\qquad
M=|F|
\]
denote the numbers of vertices and edges, respectively, in the
terminal-copy extended graph $H=(W,F)$. Each terminal is incident to
at most $n$ edges, and therefore
\[
N=O(n),
\qquad
M=O(m).
\]

The main computational cost arises from exact submodular function
minimization, which repeatedly invokes the value oracle for $f(X)$.
For any $X\subseteq U$, a fixed-flow maximum-cost flow computation
evaluates $f(X)$. Using the successive shortest augmenting path
algorithm with potentials, at most $O(n)$ augmentations are performed,
and each augmentation takes
\[
O(m+n\log n)
\]
time. Thus one value query takes
\begin{equation}
T_f(n,m)=O\bigl(n(m+n\log n)\bigr).
\label{eq:oracle-complexity}
\end{equation}

For the outer minimization, we use the exact strongly polynomial
submodular function minimization algorithm of Jiang et
al.~\cite{jiang2024convex}. On a ground set of size $n$, their
algorithm performs
\[
O(n^3\log n)
\]
value-oracle calls and requires
\[
O(n^4\log n)
\]
additional arithmetic operations. Consequently, the total running
time is
\begin{align*}
T_{\mathrm{total}}
&=O(n^4\log n)
  +O(n^3\log n)\cdot O\bigl(n(m+n\log n)\bigr)\\
&=O\bigl(n^4m\log n+n^5\log^2n\bigr).
\end{align*}

The construction of the extended graph and the final solution of the
fixed-cut linear program are polynomial-time operations and do not
change the above dominant bound.

\section{Conclusion}
\label{sec:conclusion}

We studied general edge-weighted MSPEC. Using the primal--dual
relationship for a fixed cut, we transformed the problem into the
minimization of a terminal-relaxed maximum-weight matching-cut
function. We then used a terminal-copy extended graph to prove that
this function is submodular, which yielded an exact polynomial-time
algorithm. Future work may investigate extensions with upper bounds
on vertex powers, multiple terminals, or geometric graph constraints,
as well as more efficient implementations of submodular function
minimization and network-flow routines to improve practical
performance.

\end{document}